\documentclass[%
 aip,
 jmp, 
 amsmath,amssymb,
 reprint,  
 onecolumn, 
 hyperref,
]{revtex4-1}

\usepackage{graphicx}
\usepackage{dcolumn}
\usepackage{bm}

\usepackage{graphicx,color}
\usepackage{epstopdf}
\usepackage{amsmath,amsthm,amsfonts,amssymb,times}
\usepackage{mathrsfs}

\def\be{\begin{equation}}     
\def\ee{\end{equation}}
\def\G{\mathcal{G}} 

\def\ket#1 {| #1 \rangle}
\def\Z{\mathbb{Z} }
\def\E{\mathbb E}

\numberwithin{equation}{section}

\newtheorem{theorem}{Theorem}[section]

\newtheorem{lemma}[theorem]{Lemma}

\theoremstyle{definition}

\makeatletter

\let\Im\undefined
\let\Re\undefined
\DeclareMathOperator{\Im}{Im \,}
\DeclareMathOperator{\Re}{Re \,}

\DeclareMathOperator*{\esssup}{ess\, sup \,}

\makeatother

\begin{document}
\title{Absolutely continuous spectrum implies ballistic transport for quantum particles in a random potential on tree graphs}
\author{Michael Aizenman}
\affiliation{Depts of Physics and Mathematics, Princeton University, Princeton, NJ 08544, USA}
\author{Simone Warzel}
\affiliation{Zentrum Mathematik, TU Munich, Boltzmannstr.~3, 85747 Garching, Germany}

\date{\today}

\begin{abstract} 
\hfill \emph{Dedicated to Elliott H. Lieb in celebration of his $80$th birthday} \\[2ex] 
We discuss the dynamical implications of the recent proof that for a quantum particle in a random potential on a 
regular tree graph absolutely continuous ($ac$) spectrum occurs non-perturbatively through rare fluctuation-enabled  resonances.    The main result is spelled in the title.  
\end{abstract}

\maketitle

\section{Introduction}

Progress was recently made in the understanding of the  spectra of Schr\"odinger operators with random potential on tree graphs.  In particular, it was found that absolutely continuous ($ac$) spectrum is quite robust there, and already at weak disorder ac spectrum appears  in regimes  of extremely low density of states.\cite{AWPRL,AWEuro}  The proof  suggests that in such regimes the spread of the wave function occurs by tunneling which is enabled by rare resonances.  It is  natural to ask how do  wave packets with  energies limited to such a regime spread, and at what rate does the distribution of the particle's distance from its starting point grow.    Our purpose here is to answer the latter question.  

\subsection{Bounds on quantum dynamics}

The quantum dynamics of a particle moving on a graph, which is composed of a vertex set $\mathcal{G} $ and an edge set $\mathcal{E} $, is generated by a Schr\"odinger operator of the form 
\begin{equation}\label{def:H}
	 (H\psi)(x) =  - \sum_{(x,y)\in \mathcal{E}} \psi(y) +V(x)\psi(x) \, .
\end{equation}
If the vertex degree of the graph is uniformly bounded (in which case the first term in \eqref{def:H} defines a bounded operator) then for any real valued potential $ V$ the operator $ H $ is self-adjoint operator on an appropriate domain in the Hilbert space $ \ell^2(\mathcal{G}) $.  
Under the unitary time-evolution generated by $ H $ the  probability of  
the particle having the position   $x\in \G$ at  time $ t >0 $ after it was started at a quantum state $\psi$  is given by:
\begin{equation}
	P_{\psi,t}(x) :=  \left|\left(e^{-itH}\psi\right)(x)\right|^2  \, .
\end{equation}
While the probabilistic interpretation of $P_{\psi,t}(x)$  is limited to vectors  with norm $\| \psi  \| =1$, it is convenient for us to extend the above symbol  to all $\psi$ regardless of their normalization. 

 Instead of  investigating $ P_{\psi,t}(x)  $ directly, it is often easier to study the time-averaged transition probability
\begin{equation}\label{def:pta}
	\widehat P_{\psi,\eta}(x) \ := \ 2 \eta \int_0^\infty e^{-2\eta t }\,  P_{\psi,t}(x)  \, dt \ =   \ \frac{\eta}{\pi} \int  \left|\left((H-E-i\eta)^{-1} \psi\right)(x)\right|^2 \, dE\, ,
\end{equation}
with inverse time-parameter $ \eta >0  $. 
The equality, which  is based on the spectral theorem and Plancherel's identity, links the long-time averages ($ \eta \downarrow 0$)  of the probability distribution $ \widehat P_{\psi,\eta}(\cdot) $  
with properties of the operator's Green function 
$	 G(x,y;\zeta) $
on which our analysis will focus.

The rate of growth of the distance travelled is conveniently described by the moments of the corresponding distributions, which are defined as
\begin{equation}\label{eq:mom}
	M_\psi(\beta, t) \ := \ \sum_{x\in \mathcal{G}} |x|^\beta \,  P_{\psi,t}(x) \  \quad \mbox{and}  \qquad 
	\widehat M_\psi(\beta, \eta) \ := \ \sum_{x\in \mathcal{G}} |x|^\beta \,  \widehat P_{\psi,\eta}(x)
 \, . 
\end{equation}
Here and in the following $ |x| := d(x,0) $ denotes the graph-distance of the vertex $ x $ to some fixed, but arbitrary vertex $ 0 \in \mathcal{G} $. 

To place our results in their natural context,  let us recall some points of reference  
on the relation of quantum dynamics with the spectral properties of their generator. 

\begin{enumerate}
\item  In the presence of disorder, in particular for random potentials,  there may exist subspaces of $ \ell^2(\mathcal{G}) $ over which one finds  different behavior.   For functions $\psi $ which are spanned by localized states, the moments $M_\psi(\beta, t)$ remain bounded  uniformly  in time.\cite{A92}   For functions $ \psi $ in the subspace generated by extended (generalized) eigenfunctions, the moments increase at least in the average sense that $\widehat M_\psi(\beta, \eta) \to \infty$ for $\eta \to 0$ and any $ \beta > 0 $ (by the RAGE theorem; cf.~Ref.~\onlinecite{CFKS}).

\item In the physical picture of delocalization in the presence of disorder, which was advanced by D. Thouless and collaborators, 
it is generally expected (though the statement still remains unproven)  that in the corresponding spectral regime the probability distribution  $P_\psi(x,t) $  spreads at a rate  corresponding to diffusion~(cf. Ref.~\onlinecite{YJ}).   In the finite dimensional case, of $\G = \Z^d$, that translates to: 
\be  \label{eq:mom_r}
M_\psi(\beta, t) \  \sim \ t^{r \beta} \,, 
\qquad \mbox{and}  \qquad 
	\widehat M_\psi(\beta, \eta)  \sim \ \eta^{-r \beta}  \, ,  \qquad \mbox{with $r=1/2$.}
\ee 
This is in contrast to the ballistic motion for which  $r=1$, as  is the case in the absence of disorder,  with $V$ either constant or periodic in case of $\Z^d$, or radially periodic on regular tree graphs.  
It is however relevant here to note that in the hyperbolic geometry of a regular tree the classical diffusion also spreads ballistically, since at each instance there are more directions at which  $|x|$ would increase than the one direction at which it goes down.     

 \item 
A general upper bound can be obtained from the observation  that 
on any graph $ (\mathcal{G},\mathcal{E}) $ with a uniformly bounded vertex degree, the distance travelled does not increase faster than at some finite speed $\hat v < \infty$, in the  sense that the  probability for faster growth decays exponentially: 
  \begin{equation}\label{eq:ballistic1}
{\rm Pr}_{\delta_u,t}(d(x,u) > v t) \ := \  \sum_{x: \, d(x,u) > v t} P_{\delta_u,t}(x) \   \le \  e^{-\mu t (v-\hat v)} 
\end{equation}  
 at some $\mu >0$, where the initial state  is taken to be the normalized function   localized at (an arbitrary) vertex $ u \in \mathcal{G} $, and $d(x,u)$ is the distance between the two sites $x,u\in \G$.  
This bound holds regardless of the potential $V$, and in particular it implies the ballistic upper bounds: 
\begin{equation}\label{eq:ballistic2}
      \widehat M_\psi(\beta,\eta)  \ \leq \   C(\beta,\psi) \ \eta^{-\beta} 
  \end{equation}
 for all normalized $ \psi \in \ell^2(\mathcal{G}) $ with $ 
\sum_{x\in \mathcal{G}} |x|^\beta |\psi(x)|^2 < \infty $. 
For completeness, a proof of \eqref{eq:ballistic1} is included in Appendix~\ref{App:B}, where we also comment on its relation with the Lieb-Robinson bounds.\cite{LR}     
 
\item  Lower bounds on the moments can be obtained by estimating the probability of lingering: 
\be 
{\rm Pr}_{\psi,t}(|x| < b \, t^r) \ := \  \sum_{|x| < b t^r} P_{\psi}(x,t)  \, \quad \mbox{and} \qquad 
\widehat{\rm Pr}_{\psi,\eta} (|x| < b \, \eta^{-r}) \ := \  \sum_{|x| < b \, \eta^{-r}} \widehat P_{\psi,\eta} (x,t)  \, .
\ee  
The afore-mentioned RAGE theorem implies that for any state $\psi$ within which  $H$ has only continuous spectrum:
\be  \lim_{\eta \to 0}    \widehat{\rm Pr}_{\psi,\eta} (|x| < b)   \  = \ 0 
\ee 
 for any finite $b$ (and $r=0$).  Stronger general bounds,     
due originally to I. Guarneri (with generalization found in Refs.~\onlinecite{Com93, Last96, KiLa00}), are based on the finer distinction among spectral types, classified by  the Hausdorff dimension of the spectrum and more precisely by the degree of H\"older continuity of the spectral measure associated with $\psi$.  That measure is said to be uniformly 
$ \alpha $-H\"older continuous,  for $ \alpha \in (0,1] $, if  for all  Borel sets $ I \subset \mathbb{R} $ of Lebesgue measure $ |I| \leq 1 $ one has  
	$\mu_\psi(I) \leq  C_{\psi} \, |I|^\alpha $, at a common value of $C_{\psi}$.    
For $ \mathcal{G} = \mathbb{Z}^d $,   Guarneri~\cite{G93} proved that in such cases  
$ \widehat P_{\psi,\eta}(x) \le C \, \eta^\alpha$ and thus
\be \widehat{\rm Pr}_{\psi,\eta} (|x| < b\, \eta^{-\alpha/d}) \ \le \  C_d \, b^d  \, . 
\ee 
This directly implies that  for any $\beta > 0 $ 
 \eqref{eq:mom_r} can hold only with $r\ge \alpha/d$, since selecting $b$ so that  $C_d \, b^d =1/2$ one gets:  
	\begin{equation}\label{eq:Guarneri}
		 \widehat M_\psi(\beta, \eta)   \  \geq \  \  \frac{ b^\beta }{2}  \ \eta^{-\beta \alpha/d}  \, . 
	\end{equation} 

There are  operators with absolutely continuous spectrum, corresponding to $ \alpha = 1 $, for which the Guarneri bound is almost saturated.\cite{BS00}
However, this lower bound diminishes with dimension and it provides no information for operators on tree graphs (which correspond to $d=\infty$).

\end{enumerate}

In this note we focus on  the  case  $ (\mathcal{G}, \mathcal{E} ) $ is a regular rooted tree graph and the operator~\eqref{def:H}  is random and known to have a regime of ac spectrum. 
Our main result is that the moments grow ballistically, that is  $\widehat M_\psi(\beta, \eta)$ obey not only an upper bound but also a lower bound with $r=1$. 

\subsection{Statement of the main result}

 The main topic of this note are operators of the from~\eqref{def:H} on the Hilbert space over a regular rooted tree graph, whose vertex set we denote $\mathcal{T}$,  in which every vertex aside from the root  $ 0 $ has $  K+ 1  $ neighbors with $ K \geq 2 $.  We take the potential $ V: \mathcal{T} \to \mathbb{R} $  to be random, with a distribution described by:  
  \begin{enumerate}
 	\item[A1.] $ V(x) $, $ x \in \mathcal{T} $, are independent, identically distributed random variables,
	\item[A2.]  the probability distribution of the potential at a site is absolutely continuous, $ \mathbb{P}( V(x) \in dv) = \varrho(v) \, dv $, with a bounded probability density $ \varrho \in L^\infty(\mathbb{R}) $ satisfying:
	\begin{enumerate}
	\item the moment condition with some $ r > 12 $: \; 
		$ \mathbb{E}\left[ |V(0)|^r \right] := \int_\mathbb{R} v^r \, \varrho(v) \, dv < \infty $;
	\item  the local upper bounds:  
	 \begin{equation}\label{eq:max}
	 	\varrho(v) \leq c \, \inf_{\nu \leq 1 } \frac{1}{2\nu} \int_{|v'-v| < \nu } \varrho(v') \, dv'  
	\end{equation}
	for Lebesgue-almost all $v \in\mathbb{R}$ at some $ c < \infty $.
	\end{enumerate}
\end{enumerate}

Let us recall some known facts for random Schr\"odinger operators $H$ of the form \eqref{def:H}.\cite{CL,PF,K00} By ergodicity arguments the spectrum $ \Sigma( H) $ is  almost surely given by a non-random closed set.     
Properly formulated, that also holds for the spectra corresponding to the different spectral components in the  Lebesgue decomposition of the spectral measure, i.e., the absolutely continuous ($ac$), singular continuous ($sc$), and pure point ($pp$) spectrum for which we shall use the capital letter $\Sigma_{\#}(H)$, 
with $ \# $ standing for $ac$, $sc$, or $pp$.

Instead of the closed set $ \Sigma_{ac} (H)$, we will rather focus on a measure-theoretic support of the ac density of the spectral measure.
To describe the latter, one may start from two generally valid facts: \emph{i.}~for Lebesgue almost every $E\in \mathbb R $ the limit $G(0,0;E+i0) = \lim_{\eta \downarrow 0} G(0,0;E+i\eta)$ exists almost surely, and \emph{ii.}~the $ac$ component of the spectral measure associated with the vector $\delta_0$ is  $\Im G(0,0;E+i0) \, dE / \pi$.   
We then define:
\be  \label{def:ac}
\sigma_{ac} (H) \ := \  \left\{  \, E\in \mathbb R \, : \,  \mathbb{P} \left( \Im G(0,0;E+i0) \neq 0 \right )   >  0  \ \, \right\}\,.
\ee 
For reasons which are explained  in Appendix~\ref{App:ac}, this notion of the $ac$ spectrum is better suited for out purpose than  $\Sigma_{ac}(H) $.  Both sets are non-random, and possibly differ in only in a set of zero Lebesgue measure.  However, that was not established yet.  It may be added that  at least on tree graphs, for each energy $E$ the  probability in \eqref{def:ac} is either zero or one.  It follows that  the non-random set $\sigma_{ac}(H)$ is also the support of the $ac$ component of the spectrum  for almost every realization of the randomness.     \\

While it is rather straightforward to prove that  $ \Sigma( H)  = [-2\sqrt{K},2\sqrt{K}] + {\rm supp}\, \varrho $, determining the spectral components is usually harder.  In the tree situation, the different spectral regimes can be characterized (almost completely) by a  function $ \varphi(E;1)$,  in terms of which: 
\begin{enumerate}
	\item if  $\varphi(E;1)  < \log K$ for  all energies $E$ in  some interval $I\subset \Sigma(H)$ then  $I \subset \Sigma_{pp}(H) $. 
	\item if $ \varphi(E;1)  > \log K $ for almost all energies $E$ in a measurable subset $S \subset \Sigma(H)$ then   $S\subset \sigma_{ac} (H)$  up to a difference of zero Lebesgue measure, and  furthermore  $H$ has only $ac$ spectrum in $S$.  
\end{enumerate}

The function is constructed as the boundary value  $ \varphi(E;1) := \lim_{s\uparrow 1} \varphi(E;s) $   of a large deviation free energy function defined in~\eqref{def:fe} below.\cite{AW11} 
We could not calculate $\varphi(E;1) $  explicitly, or prove regularity and continuity, except for some partial  statements.  However,  its analysis and that of the Lyapunov exponent, which bounds $ \varphi(E;1) $ from below, yields the following picture:

\begin{enumerate}
\item
In case of unbounded potentials $pp$ spectrum emerges at extreme energies ($|E|> K+1$) for small disorder\cite{A92} and it covers the whole spectrum for large disorder, as was previously indicated in the work or Abou-Chacra {\it et al}.\cite{AAT} 

\item The above criterion was recently used to show that ac spectrum emerges for  arbitrarily small disorder well beyond the spectrum of the adjacency operator, in the regime $ [-(K+1),(K+1)] \cap \Sigma(H)$ (Refs. \onlinecite{AW11,AWPRL}).   Earlier, the persistence of $ac$ spectrum  within the spectrum of the adjacency operator, i.e. $ [-2\sqrt{K},2\sqrt{K}]  \cap \Sigma_{ac}(H) \neq \emptyset $, was established by A.~Klein\cite{K}, through a continuity argument (which was recently extended to the Bethe strip~\cite{KSa}).

\item
Somewhat surprisingly, for  bounded random potentials the criterion allowed also to 
 establish that at weak disorder at its edges the spectrum is purely ac.   In particular,  in that case there is no localization at band edges.\cite{AWEuro}
\end{enumerate}

The non-perturbative  emergence of  ac spectrum well beyond the spectrum of the adjacency operator has been explained in terms of a resonance mechanism for which the 
exponential growth of the volume in a tree graph plays an important role (cf.~Refs.~\onlinecite{AW11}  and \onlinecite{AWPRL}  for a short summary). Along with that goes a picture of extended states which are localized at an infinite collection of `resonating vertices'. 

In this context, it natural to ask about the dynamical behavior of the states within the ac spectrum. It was already shown by A. Klein~\cite{K2} that in the  regime where the persistence of ac spectrum at weak disorder could be established  by a continuity argument,  the averaged dynamics is ballistic.   We now prove that  ballistic behavior extends to the full regime of ac spectrum, including the region where the dynamics appear to be dominated by tunneling events.  Following is the key bound.

 \begin{theorem}\label{thm:main0} 
Let $H$ be an operator of the form~\eqref{def:H} on a regular rooted tree graph with a random potential satisfying the above conditions~{\rm A1--2}.  Then for any initial state of the form $ \psi = f(H) \delta_0 $, with a measurable  function $ f \in L^2(\mathbb{R}) $ supported in $  \sigma_{ac}(H)$,  and all $b>0$:    
\begin{equation}\label{eq:main0}
 \mathbb{E}\left[\widehat{\rm Pr}_{\psi,\eta} (|x| < b\, \eta^{-1})\right] \ \le \  C(f)  \, b \  +  \ o(\eta)  \, . 
\ee
with some  $C(f) <\infty$, and  $o(\eta)$  a quantity which vanishes for $\eta \to 0$.  
\end{theorem}

At the risk of partial repetition, we close this section with several remarks:    
\begin{enumerate}
		\item The complementary bounds~\eqref{eq:main0} and \eqref{eq:ballistic1}  show that throughout the $ac$ spectrum the (doubly averaged) quantum time evolution in a random potential on a tree is ballistic.
	
	\item The above notwithstanding,  Theorem~\ref{thm:main0} is also consistent with the afore mentioned conjecture of diffusive evolution, since on regular tree graphs the classical diffusion spreads ballistically.   
	
\item  In extending the previous proof of ballistic behavior from the perturbative regime (of small randomness and energies within the spectrum of the adjacency operator)~\cite{K2}  to the full region of $ac$ states, Theorem~\ref{thm:main0}  excludes the possibility of another dynamical behavior in  the regime where $ac$  spectrum is caused by rare resonances.   This includes energy regimes where the density of states is extremely low,  with Lifshitz tail asymptotics (as discussed in Ref.~\onlinecite{AWPRL}).  The diffusion constant in this regime, for which the proof  yields a lower bound involving $ \mathbb{E} \left( \| f(H) \delta_0 \|^2 \right) $), should be correspondingly small.

\item One may ask whether there are operators similar to \eqref{def:H}  on tree graphs,  for which wave packets of states within the continuous spectrum spread at a slower than ballistic  rate.   Sub-ballistic rates are known to occur in classical random walks  on trees in certain `random conductance models'.  In these random walk models the Laplacian is replaced by an operator with random (though still nearest-neighbor) hopping amplitudes, whose distribution extends down to zero (cf.~Ref.~\onlinecite{Bisk}).  
	\end{enumerate}

\section{Proof of ballistic transport on tree graphs}

\subsection{A  semiclassical relation of diffusive  bound with ballistic behavior} 
Aside from some functional analytic manipulations, the  main new ingredient used here in  the proof of Theorem~\ref{thm:main0} is the statement that throughout any bounded measurable $ I \subset \sigma_{ac}(H) $ 
 the Green function's second moments obeys, for all $ x \in \mathcal{T}$: 
\begin{equation}\label{eq:mainlem}
 	\esssup_{\zeta \in I +i(0,1]} \mathbb{E}\left[|G(0,x;\zeta)|^2\right] \ \leq \frac{C_+(I)}{ K^{|x|}}  \, , 
\end{equation}
with $C_+(I) < \infty$.  (Here and in the following $ \esssup $ stands for the Lebesgue-essential supremum.)

It is instructive to note that up to a multiplicative constant, the expression of the right  in 
\eqref{eq:mainlem}  
coincides with the mean value of the total time spent at the vertex $ x $ for a particle which undergoes diffusion  originating at the root ($x=0$).  The classical expression for that is:
\begin{equation}
	\int_0^\infty \langle \delta_x , \, e^{t D} \delta_0\rangle  \, dt = \langle \delta_x \, , (-D)^{-1} \delta_0 \rangle \ = \ \frac{C}{K^{|x|} } \, ,
\end{equation}
where $ (D \psi)(x) := \sum_{y\in \mathcal{N}_x} \psi(y) - d(x) \psi(x) $ is the diffusion generator  (with $\mathcal{N}_x$ the collection of sites neighboring $x$ and $d(x) = | \mathcal{N}_x|$ the site's degree, which is $K$ at the root and $(K+1)$ elsewhere). 
A comparison of~\eqref{eq:mainlem} with~\eqref{def:pta} reveals that such a bound  may indeed be expected for the Green function's disorder-averaged second moment  if over the range of energies $ I $ the dynamics  is  \emph{at least diffusive}.   However one should bear in mind that on trees the distance of a diffusing particle from its initial point grows ballistically.    Our bounds on the quantum evolution  bear out this  double perspective.

In Lemma~\ref{lem:fv} we prove that  \eqref{eq:mainlem}  is implied by the statement that 
for any bounded subset $I \subset \sigma_{ac}(H)$:  
\begin{equation} \label{eq:delta}
	\esssup_{\zeta \in I +i(0,1]}  \mathbb{E}\left[(\Im G(0,0;\zeta))^{-3-\delta}\right]   \ < \ \infty . 
\end{equation}
at some $\delta >0$ (which does not depend of $I$).    The derivation of \eqref{eq:delta}  is the subject of Theorem~\ref{thm:Fprop}.  

Assuming the validity of \eqref{eq:mainlem} the proof of Theorem~\ref{thm:main0} is rather elementary, while  the proofs of~Lemma~\ref{lem:fv}  and Theorem~\ref{thm:Fprop} are increasingly more involved.  We shall therefore establish the above statements in this order: first show how  \eqref{eq:mainlem} implies the bound claimed in  Theorem~\ref{thm:main0}, then through Theorem~\ref{lem:fv}  show how  \eqref{eq:delta} implies \eqref{eq:mainlem}, and finally  (in part C) establish Theorem~\ref{thm:Fprop} and through it \eqref{eq:delta}.   

\subsection{Conditional proof of the main result}

We start by relating $ \mathbb{E}[\widehat P_{f(H)\varphi, \eta}] $ with an even more convenient quantity.   The difference between the two is the  vanishing term, $o(\eta)$, in~\eqref{eq:main0}. As a technical tool, we will employ the Wegner estimate which guarantees the absolute continuity of the average of the spectral measure $ \mu_\varphi $ of the random operator associated with any (a priori fixed) vector $ \varphi \in \ell^2(\mathcal{G})$:
\be  \label{eq:Wegner}
\E[\mu_\varphi(dE) ] \le C_W  \| \varphi \|^2 \, dE \, , 
\ee 
where $ C_W $ is a finite constant.\cite{W,Sim_Wolff}

\begin{lemma}\label{lem:enout}   
For a random operator satisfying the Wegner estimate~\eqref{eq:Wegner}, 
any $ f \in L^2(\mathbb{R}) $ and any $ \varphi \in \ell^2(\mathcal{G}) $ let 
\begin{equation}
		K_{\varphi,f,\eta}(x) \ := \ \frac{\eta}{\pi} \int |f(E)|^2 \left| \left((H-E-i\eta)^{-1} \varphi\right)(x)\right|^2 dE  \, . 
	\end{equation}	
Then the following $ \ell^1 $-convergence holds 
\begin{equation}
		\lim_{\eta\downarrow 0} \sum_{x\in\mathcal{G}} \left| \,  
		\E [ \widehat P_{f(H)\varphi}(x;\eta) ] - \E [K_{\varphi,f,\eta}(x) ] \right| \ = \ 0 \, .
	\end{equation} 
\end{lemma}
\begin{proof}
The Wegner's estimate~\eqref{eq:Wegner} guarantees the uniform boundedness of the $ \ell^1 $-norm of $K_{\varphi,f,\eta} $. Namely, abbreviating $ \delta_\eta(x) := \pi^{-1} \Im (x-i\eta)^{-1} $, which is an approximate $\delta$-function, the spectral representation yields:
\begin{equation}
	0 \leq \sum_{x \in \mathcal{G}} K_{\varphi,f,\eta}(x)  = \int |f(E)|^2 \int \delta_\eta(E'-E) \, \E[\mu_\varphi(dE') ]   \, dE \ \leq \ C_W  \| \varphi \|^2  \, \| f \|^2 \, . 
\end{equation}
Using~\eqref{def:pta}, the triangle inequality and the Cauchy-Schwarz inequality, it is not hard to see that
\begin{equation}\label{eq:L2}
	 \sum_{x\in\mathcal{G}} \left| \,  
		\E [ \widehat P_{f(H)\varphi,\eta}(x) ] - \E [K_{\varphi,f,\eta}(x) ] \right| \ \le  \ Q(\eta) + 2 \sqrt{C_W}  \| \varphi \|  \, \| f \| \,  \sqrt{Q(\eta)} \, ,  
\end{equation}
where
\begin{align}
Q(\eta) \ := &  \   \frac{\eta}{\pi} \sum_{x \in \mathcal{G}} \int \E\left[ 
\left| \left((H-E-i\eta)^{-1} f(H) \varphi\right)(x) - f(E) \, \left((H-E-i\eta)^{-1} \varphi \right)(x)  \right|^2 \, \right] \, dE \notag \\
 = & \ \E\left[ \int \int \delta_\eta(E'-E) \left| f(E') - f(E) \right|^2  \, dE \, \mu_\varphi(dE') \right]  \, . 
\end{align}
The equality is again based on the spectral representation. Applying now the Wegner bound~\eqref{eq:Wegner} again, we get: 
\begin{eqnarray} 
Q(\eta)  \ & \le & \ C_W\, \| \varphi \|^2\,   \int \int\delta_\eta(E'-E) \left| f(E') - f(E) \right|^2 dE \, dE'   \nonumber \\[2ex] 
&=& 
 \ 2\, C_W\, \| \varphi \|^2\,   \left[ \int   |f(E)|^2 dE -  \int_\mathbb{R}
 \int  \bar f(E')  \ \delta_\eta(E'-E) \   f(E)  \, dE \, dE'  \right]
 \,    . 
\end{eqnarray}  
In the limit  $\eta \downarrow 0$ the above quantity vanishes due to the weak convergence to identity of the operator in $L^2(\mathbb{R})$ whose kernel is $\delta_\eta(E'-E)$.  The latter statement is easily seen in the Fourier representation, where the operator corresponds to  multiplication by $e^{-|\tau|\eta}$, with $\tau$ denoting the Fourier transform variable. 

\end{proof}

Assuming now the bound \eqref{eq:mainlem}, which is proven below independently of the next  argument,  we proceed to the first of the three steps outlined above. 

\begin{proof}[Conditional proof of Theorem~\ref{thm:main0}]
By Lemma~\ref{lem:enout} (and using the notation introduced in its proof) 
\begin{equation}  \label{eq:twidle1}
  \Big|   \mathbb{E}\left[
  \widehat{\rm Pr}_{\psi,\eta} (|x| < R)
  \right]  -  \sum_{x:\, |x|<R} \mathbb{E}\left[K_{\delta_0,f,\eta}(x)  \right]  \Big| \ = \ o(\eta) \, . 
	\end{equation}
 In the special case $ \varphi = \delta_0 $ one has  
\begin{equation}\label{eq:secPtilde}
		K_{\delta_0,f,\eta}(x) \  = \  \frac{\eta}{\pi} \int |f(E)|^2 \left| G(x,0;E+i\eta) \right|^2 dE  \, . 
\end{equation}
The estimate \eqref{eq:mainlem} with $ I_f\ :=\ \{E\in \mathbb R :  f(E) \neq 0 \}  $ yields
\begin{align}
	\eta \sum_{|x| < R} \mathbb{E}\left[|G(x,0;E+i\eta)|^2\right] \ & 
	 \  \leq \  C_+( I_f ) \, \eta \, \sum_{n=0}^{ R -1} 1\  = \ C_+( I_f) \, \eta \, R \, . 
\end{align}	
We thus have 
\begin{align} \label{eq:twidle2}
\sum_{|x| < R }  \mathbb{E}\left[K_{\delta_0,f,\eta}(x)\right]    \leq \frac{C_+( I_f)   \, \|f\|_2^2}{\pi} \; \eta \, R \, , 
  \end{align}
The proof of Theorem~\ref{thm:main0} is concluded by combing the bounds 
\eqref{eq:twidle1}  and \eqref{eq:twidle2}, and 
choosing  $R=b \ \eta^{-1}$. 
\end{proof}

\subsection{The utility 
of the negative moments of the Green function}

Our next goal is to show that \eqref{eq:mainlem} follows from \eqref{eq:delta}.  In the proof we shall make use of some of the structure which was developed in Ref.~\onlinecite{AW11}.  
It was proven there that the following limit, which defines what is called  there the  free-energy function,  exists and is finite for any $ s \in [0,\infty) $ and  $ \zeta \in \mathbb{C}^+ := \{ z\in \mathbb{C}\, : \, \Im z >0\} $
\begin{equation}\label{def:fe}
	\varphi(s;\zeta) \ := \ \lim_{|x|\to \infty} \, \frac{1}{|x|} \log \mathbb{E}\left[ |G(0,x;\zeta)|^s \right]  \, . 
\end{equation}
Some  useful properties (taken from Section~3 in  Ref.~\onlinecite{AW11}) are:
\begin{enumerate}
	\item For any $ s \in [0,2] $ and any $ \zeta \in \mathbb{C}^+ $:
	\begin{equation}\label{eq:apriori}
		\varphi(s;\zeta)  \ \leq \ - s \, \log K \, .
		\end{equation}
	In fact, the inequality is strict for any $ 	\zeta \in \mathbb{C}^+ $.
	\item For any $ s \in [0,1) $ and $x\in \mathcal{T}$ the following `finite-volume bounds' hold 
 \begin{equation}\label{eq:finitevolume}
 	C_-(s;\zeta)   \, e^{\varphi(s;\zeta) \, |x| }  \ \leq \  \mathbb{E}\left[\left|G(0,x;\zeta)\right|^s\right] \ \leq \ C_+(s;\zeta)\,  e^{\varphi(s;\zeta) \, |x| } 
 \end{equation}
with $ C_\pm(s;\zeta)  \in (0,\infty) $, which  at fixed $ s \in [0,1) $ 
are bounded uniformly in  $  \zeta  \in [-E,E] +i  (0,1] $ for any   $  0 \le E < \infty $.  
\end{enumerate} 

In general one does not expect the bounds~\eqref{eq:finitevolume} to hold beyond $ s = 1 $,  since, for instance, for energies in the regime of pure point spectrum:  $\mathbb{E}\left[\left|G(0,x;E+i0)\right| \right] = \infty $ . However, as we assert next (based on the argument provided in Ref.~\onlinecite{AW11}) these bounds do extend to all energies at which the imaginary part of the resolvent has a finite inverse moment of power greater than one.   As will be shown in Theorem~\ref{thm:Fprop} below, this includes the entire $ac$ spectrum. 
\begin{theorem}\label{lem:fv}
Under the assumptions of Theorem~\ref{thm:main0}, if  for some bounded measurable set $ I \subset \mathbb{R} $ and some $ \delta > 0 $ 
\begin{equation}\label{eq:invmom}
	\esssup_{\zeta \in I +i  (0,1]} \mathbb{E}\left[ \left( \Im G(0,0;\zeta) \right)^{-3-\delta} \right] \ < \ \infty \, ,
\end{equation} 
then for almost all $ E \in I $ and all $ \eta \in (0,1] $: 
\begin{equation}\label{eq:fv}
 C_-\,  e^{\varphi(2;E+i\eta) |x|}  \ \leq \ \mathbb{E}\left[ |G(0,x;E+i\eta)|^2 \right] \ \leq  \ C_+\,  e^{\varphi(2;E+i\eta) |x|} 
 \end{equation} with some   $ C_\pm \in (0,\infty) $. 
\end{theorem}
Of main interest for us  is the \emph{upper bound} in~\eqref{eq:fv}, which together with~\eqref{eq:apriori} yields~\eqref{eq:mainlem}.  \\ 

For the proof of Theorem~\ref{lem:fv}, which proceeds essentially along the lines of Theorem~3.2 in Ref.~\onlinecite{AW11}, 
we recall some special properties of the Green function in a tree geometry: 
\begin{enumerate}
	\item For any $ x \in \mathcal{T} \backslash \{ 0 \} $ we denote by $\mathcal{P}_{0,x} $ 
the unique path  connecting $ 0 $ and $ x $.  	
For  each vertex $u\in \mathcal{P}_{0,x}$  other than the path's endpoints $0$ and $x$,   the vertex set $\mathcal{T}_u := \mathcal{T} \backslash \{ u \} $ decomposes into $ K+1 $ trees rooted at the neighbors $ v \in \mathcal{N}_u $ of $ u $.  We denote by $ G^{\mathcal{T}_{u}} $ the Green function of the natural restriction of $H$ to $ \ell^2(\mathcal{T}_u) $. 		
	The diagonal element of the Green function at $ u $ can then be written as
	\begin{equation} \label{eq:Grep}
	G(u,u; \zeta) \ = \ \Big(  V(u)  - \zeta -  \sum_{v\in {\mathcal N}_u} G^{\mathcal{T}_{u}}(v,v;\zeta) \Big)^{-1} \, ,
	\end{equation}
	and the Green function between $0$ and $x$  admits the factorization:
	\begin{equation}\label{eq:factG}
	G(0,x;\zeta) = G^{\mathcal{T}_{u}}(0,u_-;\zeta) \, G(u,u;\zeta) \, G^{\mathcal{T}_{u}}(u_+,x;\zeta) \, ,
\end{equation}
	where  $ u_\pm $ is the forward/backward neighbor of $ u $ on the path $ \mathcal{P}_{0,x} $.  
	\item Conditioning on the sigma algebra $\mathscr{A}_{u} $ generated by the random variables $ \{ V(y) \, | \, y \in  \mathcal{T}_u \} $, we thus obtain: 
	\begin{align}\label{eq:condE}
		\mathbb{E}\left[ \left| G(u,u; \zeta) \right|^{2+\beta} \big| \mathscr{A}_{u}\right] \ & \leq \ \| \varrho \|_\infty \, \mathbb{E}\Big[ \int \Big| v - \zeta -   \sum_{v\in {\mathcal N}_u} G^{\mathcal{T}_{u}}(v,v;\zeta) \Big|^{-2-\beta} dv \Big] \notag \\ 
			& \leq \  C_\beta \, \| \varrho \|_\infty \; \mathbb{E}\left[ \left( \Im G^{\mathcal{T}_{u}}(v,v;\zeta) \right)^{-1-\beta} \right] 
	\end{align}
	where $ \beta \geq 0 $ and $ v \in {\mathcal N}_u \backslash \{ u_- \} $. 
	Note that in this case $ G^{\mathcal{T}_{u}}(v,v;\zeta)  $ has the same distribution as $ G(0,0;\zeta) $.
\end{enumerate}
\begin{proof}[Proof of Theorem~\ref{lem:fv}]
As explained in Ref.~\onlinecite{AW11}, the claim follows from sub-/supermultiplicative bounds of the form:
\begin{equation} \label{eq:surface}
c_-^{-1} \ \le \   \mathbb{E}_2\left[| G(u,u; \zeta) |^2 \right] 
   \ \le \  c_+ \, ,
\end{equation}
for all $ \Re \zeta \in I $ and $ \Im \zeta \in (0,1] $, where we use the abbreviation
\begin{equation}\label{eq:defE2}
	   \mathbb{E}_2\left[\cdot\right] \ := \  \frac{ \mathbb{E}\left[ |G^{\mathcal{T}_{u}}(0,u_-; \zeta)|^2 \,  |G^{\mathcal{T}_{u}}(u_+,x_-; \zeta)|^2 \,   (\cdot)\,\right] } {   \mathbb{E}\left[|G^{\mathcal{T}_{u}}(0,u_-; \zeta)|^2\right] \,  \mathbb{E}\left[ |G^{\mathcal{T}_{u}}(u_+,x_-; \zeta)|^2   \right]} 
\end{equation}
for a  weighted, or \emph{tilted} - in the language of large deviations theory, expectation. [Note that the latter depends on various parameters such as $ \zeta $ and the involved vertices which are suppressed in the notation.] 

The claimed bounds~\eqref{eq:surface} rely on the factorization property~\eqref{eq:factG}. 
Namely, by virtue of~\eqref{eq:Grep} and \eqref{eq:condE} (with $ \beta = 0 $), the upper bound holds  with $ c_+ =  \pi \, \| \varrho \|_\infty \, \esssup_{\zeta \in I +i  (0,1]} \mathbb{E}\left[ \left( \Im G(0,0;\zeta) \right)^{-1} \right]   $.

The proof of the lower bound rests on~\eqref{eq:factG}  and~\eqref{eq:Grep} which yields for any $ t > 0 $:
\begin{equation}
 \mathbb{E}_2\left[\left| G(u,u; \zeta) \right|^2 \right]  \ \geq \ \int \frac{ \varrho(v) \, dv}{ \left[ | v| + |\zeta| + (K+1) t \right]^{2} } \, \prod_{w \in \mathcal{N}_u}   \mathbb{P}_2\left(  | G^{\mathcal{T}_{u}}(w,w; \zeta) | \leq t  \right) \, ,
\end{equation}
where we have used the independence (also under the titled measure) of the Green functions $ G^{\mathcal{T}_{u}}(w,w;\zeta) $ for different $ w \in \mathcal{N}_u $.  
The probability in the right side is estimated with the help of a Chebychev bound
\begin{equation}\label{eq:cheb}
 \mathbb{P}_2\left(  | G^{\mathcal{T}_{u}}(w,w; \zeta) | \leq t  \right) \ \ \geq \ 1-  \frac{ \mathbb{E}_2\left[  | G^{\mathcal{T}_{u}}(w,w; \zeta) |^s\right]}{t^{s}}  \, .
 \end{equation}
For $ w \in \mathcal{N}_u \backslash \{ u_+ \, , \, u_- \} $, the tilded measure coincides with the original one and one has the uniform estimate for $ s \in (0,1) $:
\begin{equation}
 \mathbb{E}\left[  | G^{\mathcal{T}_{u}}(w,w; \zeta) |^s\right] \ \leq \  \|\varrho \|_\infty \, \sup_{\zeta \in \mathbb{C}^+ } \int  \frac{dv}{ |v-\zeta|^{s}}  \quad \left[ =: \ D_s \right] \, . 
\end{equation}
For $ w \in \{ u_+ \, , \, u_- \} $ we need a separate argument to prove that the quantity $ \esssup_{\zeta \in I +i  (0,1]}  \mathbb{E}_2\left[  | G^{\mathcal{T}_{u}}(w,w; \zeta) |^s\right]  $ is  also bounded below 
uniformly in $ x $ and $ u $.  That is spelled below in Lemma~\ref{lem:FMtilt}.  We thus conclude that there exists $ t > 0$ such that for all  $ x $ and $ u $ the right side in~\eqref{eq:cheb} is larger that $ 1/2 $.  The claimed lower bound then readily follows.  
\end{proof}

In the above proof we encountered tilted expectation values of random variables $ G^{\mathcal{T}'}(x_0,x_0;\zeta) $, where $  \mathcal{T}' $ is either isomorphic to $ \mathcal{T} $ or of the form $ \mathcal{T}_v $ and $ x_0 $ is the root in $  \mathcal{T}' $.   The tilting is along a finite segment of an infinite path of vertices which in each case we list as $  x_0, \dots , x_n, \dots  $ in $  \mathcal{T}' $: 
\begin{equation}
\mathbb{E}_{2}^{(0,n)}\left[\cdot \right] := \frac{\mathbb{E}\left[|G^{\mathcal{T}'}(x_0,x_n;\zeta)|^2 \left( \cdot \right) \right] }{\mathbb{E}\left[|G^{\mathcal{T}'}(x_0,x_n;\zeta)|^2  \right]} \, . 
\end{equation}
The auxiliary result quoted there is that these quantities are  bounded independently of $ n $.  Following is the derivation of that statement (when applying it to~\eqref{eq:cheb} $x_0$ is to be adjusted to the relevant point $u_\pm $.) 
\begin{lemma}\label{lem:FMtilt}
Assuming~\eqref{eq:invmom}, and the  assumptions on $H$ which are spelled in Theorem~\ref{thm:main0}, 
for any $ s \in(0,1)$ there is  $ C_s(I) < \infty$ such that the following bound holds uniformly   for all $ n \in \mathbb{N}_0 $:
\begin{equation} \label{eq:II3}
R_n\ := \  \esssup_{	\zeta \in I +i  (0,1]} \mathbb{E}_{2}^{(0,n)}\left[\left|G^{\mathcal{T}'}(x_0,x_0;\zeta) \right|^s \right] \ \leq \  C_s(I) \, .  
\end{equation}
\end{lemma}
\begin{proof}
	 In the argument use will repeatedly be made of the following  general interpolation bound.   For any random variable $Y$ the moment function $M(s) := \E(|Y|^s)$ (for which $M(0)=1$) satisfies, for all $s, \alpha \ge 0$:
\be \label{conbound} 
\frac{M(2+s)}{M(2)} \ \le \ \left[   M(2+s+\alpha)^s\, M(-s)^2 
\right]^{\frac{1}{s + \alpha}  } \, . 
\ee
This relation is implied by the well known convexity of the function $s\mapsto \log M(s)$:
\be \label{covex1} 
\frac{\log M(2+s) - \log M(2)}{s}  \  \le \   \frac{\log M(2+s + \alpha) - \log M(2)}{s + \alpha} \    
\le \   \frac{\log M(2+s + \alpha)}{s + \alpha} 
+ \frac{2 \log M(-s)}{s (s + \alpha)}
\ee 
where the first inequality is by the monotonicity relation of the slopes of the function's cords, and the second employs the bound   
$
\log M(2) \ \ge \  - 2 \log M(-s) / s $,  
which expresses the convexity relation among the values of: 
$\log M(-s),  \log M(0) = 0$ and $\log M(2)$. 
 The  claimed bound \eqref{conbound} is obtained by the exponentiation of~\eqref{covex1}.   The proof of the bound \eqref{eq:II3} 
will proceed now by induction on $ n \in \mathbb{N}_0$.

For  $ n = 0 $ let us note that the relevant tilted expectation value takes the form of the ratio of two moments: 
\be  \label{Mrep}
\mathbb{E}_{2}^{(0,0)}\left[\left|G^{\mathcal{T}'}(x_0,x_0;\zeta) \right|^s \right]  = \frac{M(2+s)}{M(2)} \, , 
\ee 
as in  \eqref{conbound}, of the quantity
	\begin{equation}\label{eq:reprep}
	 Y := G^{\mathcal{T}'}(x_0,x_0;\zeta) \ = \  \Big(  V(x_0)  - \zeta -  \sum_{v\in {\mathcal N}_{x_0}} G^{\mathcal{T}_{x_0}'}(v,v;\zeta) \Big)^{-1}
	 \end{equation}
where the last equality is by~\eqref{eq:Grep}.   
	 For the desired  bound, we  apply \eqref{conbound} with 
 $ \alpha = 2 - s + \delta $.   Since in \eqref{eq:reprep} there is at least is one neighbor $ v \in \mathcal{N}_{x_0} $ such that $ G^{\mathcal{T}_{x_0}'}(v,v;\zeta)  $ is identically distributed with $ G(0,0;\zeta) $, a calculation similar to  \eqref{eq:condE} (with $ \beta = 2+\delta $) yields
	\begin{equation}
	M(2+s + \alpha) \ = \ \mathbb{E} \left[ \mathbb{E}\left(  \left|Y\right|^{4+\delta} | \mathcal{A}_{x_0} \right) \right]  \ \leq \ C_{2+\delta} \, \| \varrho \|_\infty \, \mathbb{E}\left[ \left( \Im G(0,0;\zeta) \right)^{-3-\delta} \right]  
	\end{equation}
	which  by the explicitly assumed condition~\eqref{eq:invmom}  is uniformly bounded for almost all $ \zeta \in I +i(0,1] $. 
	
	For the fractional moment, using~\eqref{eq:reprep} and the triangle inequality one gets
	\begin{equation}
	M(-s) \ = \ 	 \mathbb{E}\left[ \left|Y \right|^{-s} \right] \ \leq \ \int \varrho(v) |v|^s dv + |\zeta|^s + K \, D_s \, .  
	\end{equation}
Applying now \eqref{conbound} to the expression in \eqref{Mrep} we conclude that $R_0<\infty$. 
	
	For the induction step ($ n-1 \to n  $), equation \eqref{Mrep} is replaced by:  
	\begin{equation}
	\mathbb{E}_{2}^{(0,n)}\left[\left|G^{\mathcal{T}'}(x_0,x_0;\zeta) \right|^s \right]  = \frac{\mathbb{E}_{2}^{(1,n)}\left[Y^{2+s} \right]}{ \mathbb{E}_{2}^{(1,n)}\left[Y^{2} \right]} \, . 
	\end{equation}
The interpolation relation \eqref{conbound} with $ \alpha = 2 - s + \delta $ bounds the right side by two terms. The first  is estimated 
similarly to~\eqref{eq:condE}:
	\begin{equation}
		\mathbb{E}_{2}^{(1,n)}\left[Y^{4 + \delta} | \mathcal{A}_{x_0} \right] \ \leq \ \| \varrho\|_\infty  \int \Big| v - \zeta - \sum_{v \in \mathcal{N}_{x_0}} G^{\mathcal{T}_{x_0}}(v,v;\zeta)   \Big|^{-4- \delta}\,  dv  \ \leq \ C_{2+\delta} \, \| \varrho \|_\infty \, \mathbb{E}\left[ \left( \Im G(0,0;\zeta) \right)^{-3-\delta} \right]   \, .
	\end{equation}
	Here the last inequality relies again on the fact that at least is one neighbor  of $ x_0 $ is identically distributed with $ G(0,0;\zeta) $. The second term
	is 
	\begin{align}
		\mathbb{E}_{2}^{(1,n)}\left[Y^{-s} \right] \ &  \leq \  \int \varrho(v) |v|^s dv + |\zeta|^s  + \sum_{v \in \mathcal{N}_{x_0}} \mathbb{E}_{2}^{(1,n)}\left[ |G^{\mathcal{T}_{x_0}}(v,v;\zeta)|^s \right] \\
		& \leq \ \int \varrho(v) |v|^s dv + |\zeta|^s + (K-1) D_s  + \mathbb{E}_{2}^{(1,n)}\left[\left|G^{\mathcal{T}_{x_0}}(x_1,x_1;\zeta) \right|^s \right] \, . 
	\end{align}
	In summary, we have thus established that there are constants $ A_s , B_s $, with which:
	\begin{equation} \label{esssup}
		R_{n} 
		\ \leq \ A_s \left( B_s + R_{n-1} \right)^{\frac{2}{2+\delta}} \, . 
	\end{equation}
It then follows inductively  that 
for all $n\ge 0$:  
\be  R_n \ \le \  C_s(I) \ := \  \min \{ u \ge R_0:  \,  A_s \left( B_s +u \right)^{\frac{2}{2+\delta}}< u\} \, ,  
\ee 
 which is finite since $ R_0 <\infty$ and $\frac{2}{2+\delta}< 1$.  
	This concludes the proof.
\end{proof}

\subsection{Finiteness of the Green function's inverse moments}

In order to apply Lemma~\ref{lem:fv} to  the proof of our main result one needs to establish the finiteness of the Green function's inverse moments  for energies within the $ac$ spectrum,  as expressed in~\eqref{eq:invmom}.   The starting point for that is the relation: 
	\begin{equation}\label{eq:recrel}
		\Im G(0,0;\zeta) \ \geq \ | G(0,0;\zeta) |^2 \, \sum_{v \in \mathcal{N}_0} \Im G^{\mathcal{T}_0}(v,v;\zeta) \, ,
	\end{equation}
which follows from~\eqref{eq:Grep}.  Among its implications is the zero-one law which was noted and applied in Ref.~\onlinecite{AW11}: 
 for each energy $E\in S$ (the full measure subset of $\mathbb R$ over which $G(0,0;E+i0)$  is defined), 
 $\mathbb{P}\left( \Im G(0,0;E+i0) \neq 0\right) $ is either $0$ or $1$.    Equivalently, for any $E\in \sigma_{ac}(H)$, as defined by \eqref{def:ac}: 
  $ \mathbb{P}\left( \Im G(0,0;E+i0) \leq x \right)  \to 0 $, as $x\to 0$.   The following may be viewed as a  quantitative improvement (though not yet the best possible) on that statement.  
 
 \begin{theorem}\label{thm:Fprop}  Under the assumptions of Theorem~\ref{thm:main0}, 
for any bounded and measurable set  $ I\subset \sigma_{ac}(H) $ the function  
	\begin{equation} \label{eq:x}
		F(x) \ := \ \esssup_{\zeta \in I+i(0,1]} \mathbb{P}\left( \Im G(0,0;\zeta) \leq x \right)  
	\end{equation}
satisfies: 
\begin{equation}\label{eq:boundF}
	F( x ) \leq C x^{3+\epsilon}
\end{equation}
for all $ x \in [0,x_0] $, at some (finite) $ C , x_0, \epsilon > 0 $. 
\end{theorem} 

Before proving it, let us note that Theorem~\ref{thm:Fprop} ensures the validity of ~\eqref{eq:invmom}, since \eqref{eq:boundF} implies that for all $ \delta \in (0,\epsilon) $:
  \begin{equation}
	\esssup_{\zeta \in I +i(0,1]}  \mathbb{E}\left[(\Im G(0,0;\zeta))^{-3-\delta}\right] \ \leq \ \ (3+\delta) \int_0^\infty \frac{F(x)}{x^{4+\delta}} \, dx  \ < \ \infty . 
\end{equation}  
  We did not  push the limits here. In fact, in case $ {\rm supp} \,  \varrho $ is compact, the subsequent proof shows that $ F(x) \leq C x^\gamma $ for any $ \gamma > 0 $.  

The proof utilizes the observation that by \eqref{eq:recrel}  the small probability event $ \{  \Im G(0,0;E+i0) \leq x \} $ requires the occurrence of  $K$ similar and somewhat uncorrelated events $\{ | G(0,0;\zeta) |^2 \,  \Im G^{\mathcal{T}_0}(v,v;\zeta)   \leq x \}$.  Had such a relation been valid with  $ | G(0,0;\zeta) |^2 $ replaced by a positive constant, it would easily follow that the probability $F(x)$ vanishes faster than any power of $x$.   However, as it is,  $ | G(0,0;\zeta) |^2 $ can be arbitrarily small.   Taking that into account, we get the following relation.
 
\begin{lemma}\label{lem:Fprop}  The function $ F(x)$ defined in  Theorem~\ref{lem:Fprop}  is monontone increasing in $x$, satisfies $ \lim_{x\downarrow 0} F(x) = 0 $, and 
		\begin{equation}\label{eq:nonlin}
			F(x) \ \leq \ F(x y^{-2})^K + C  \left(y \, F(c y)^K + y^{r} \right) \, ,  
		\end{equation}

for all $ x  >0 $ and  $ y \in (0,y_0) $, with $r= 6$ and some constants $ c, C, y_0  \in (0,\infty) $. 
\end{lemma}

\begin{proof} The first statement is implied by  monotonicity and the above mentioned $0$-$1$ law (Lemma~4.1 of Ref.~\onlinecite{AW11}). 
To arrive at \eqref{eq:nonlin} we note that for any $y>0$ the event in \eqref{eq:x} occurs only if either 
$| G(0,0;\zeta) | \leq y$ or for all $v \in \mathcal{N}_0$: $\Im G^{\mathcal{T}_0}(v,v;\zeta) \le xy^{-2}$.   
	Introducing the distribution functions $ H_\zeta(y) \ := \ \mathbb{P}\left(| G(0,0;\zeta) | \leq y \right) $ and $ F_\zeta(x) := \mathbb{P}\left( \Im G(0,0;\zeta) \leq x \right) $,  we   conclude that for all $ y > 0 $:
\begin{equation}\label{eq:step1F}
	F_\zeta(x) \ \leq \ F_\zeta(x y^{-2} )^K +H_\zeta(y) \, .
\end{equation}
To `close' this relation we employ the bound on $H_\zeta(y) $ which is presented in the Lemma~\ref{lem:auxH}, below.  Using it, the relation \eqref{eq:step1F}  yields for $ F_\zeta$ (and hence also for $ F $)  the following non-linear inequality:
 \begin{equation}
	F_\zeta(x)  \ \leq \ F_\zeta(x \, y^{-2})^K \ + \  4 K^2 \, \| \varrho \|_\infty \, y \, F_\zeta(2 K y)^K \ + \ \mathbb{P}\left( |V(0)| \geq (4y)^{-1}\right) 
\end{equation}
for any $ y \in (0, (4|\zeta|)^{-1}) $. Applying to that the Chebychev bound $ \mathbb{P}\left( |V(0)| \geq y^{-1}\right) \leq y^r\,  \mathbb{E}\left[|V(0)|^r\right] $, and the finiteness of $ \sup_{\zeta \in I+i(0,1]} |\zeta|$, one may deduce \eqref{eq:nonlin}. \\ 
\end{proof} 

In the above proof we relied on the following auxiliary statement. 
\begin{lemma}\label{lem:auxH}  Under the assumptions of Theorem~\ref{thm:main0}, for all $ x \in (0,\infty) $ one has:       
\begin{enumerate}    
	\item $     
\displaystyle
  H_\zeta(x) \ \leq \  \mathbb{P}\left( |V(0)| \geq (4x)^{-1}\right) + K \, \mathbb{P}\left(  | G(0,0;\zeta)  |    \geq    (2 K x)^{-1} \right) $ provided $ |\zeta| \leq (4x)^{-1} $.
  \item $ 
\displaystyle
 1 -  H_\zeta(x^{-1}) \leq  2 \| \varrho \|_\infty \, x\, F_\zeta(x)^K  $.
  \end{enumerate}
\end{lemma}
\begin{proof}
	We use the representation~\eqref{eq:Grep} to conclude:
	\begin{align}
		H_\zeta(x)  \ & \leq \ \mathbb{P}\Big(|V(0) | + |\zeta|  + \sum_{v \in \mathcal{N}_0} | G^{\mathcal{T}_0}(v,v;\zeta) | \geq x^{-1}  \Big)  \leq \ \mathbb{P}\left( |V(0)| \geq (4x)^{-1}\right) + K \, \mathbb{P}\left(  | G(0,0;\zeta)  |    \geq    (2 K x)^{-1} \right) \, .
	\end{align}
	Here the second inequality relied on $ |\zeta| \leq (4x)^{-1} $ and the fact that $    G^{\mathcal{T}_0}(v,v;\zeta)   $ with $ v \in \mathcal{N}_0 $ is identically distributed as $ G(0,0;\zeta) $. 
	Employing~\eqref{eq:Grep} again we get:
	\begin{align}
		1 - H_\zeta(x^{-1}) & = \mathbb{P}\left(  | G(0,0;\zeta)  |    >    x^{-1} \right)  \notag \\
		 & \leq  \ \mathbb{P}\Big( \big|V(0) - \Re \zeta -  \sum_{v \in \mathcal{N}_0} \Re G^{\mathcal{T}_0}(v,v;\zeta)  \big| \leq x \; \mbox{and} \;  \sum_{v \in \mathcal{N}_0} \Im G^{\mathcal{T}_0}(v,v;\zeta) < x \Big) \notag \\
		& \leq    \    2 \, \| \varrho \|_\infty \, x \;  \mathbb{P}\Big( \sum_{v \in \mathcal{N}_0} \Im G^{\mathcal{T}_0}(v,v;\zeta) \leq x \Big)  \ \leq \   2\, \| \varrho \|_\infty  \,  x  \, F_\zeta(x)^K \, , 
	\end{align}
where  the last step is due to the independence of the variables $ G^{\mathcal{T}_0}(v,v;\zeta) $, $ v \in \mathcal{N}_0 $. 
\end{proof}

We  now turn to the main result of this section. 

\begin{proof}[Proof of Theorem~\ref{thm:Fprop}] The proof proceeds in three steps. We  first establish  that a power-law bound of the form~\eqref{eq:boundF} holds with at least a small power.      Next, the power is improved to a value greater than $1$.  That step is simpler for $K$ larger than $2$.  The remaining case $K=2$ requires an extra argument, which forms the third step in the proof.   

For an initial power bound, we pick  $ y = x^{1/4} $ in \eqref{eq:nonlin} to conclude that for all $ x \in (0,\infty) $
\begin{equation}\label{eq:boundcorse}
 F(x^2) \ \leq 3 \, \max\left\{ F(x)^K \, , \, C\,  \sqrt{x} \right\} \, . 
\end{equation}
The convergence $ \lim_{x\downarrow 0 } F(x) = 0 $ implies that there is some $ x_0 \in (0, \min\{ (9C)^{-4}, 1/2\} ] $ such that $F(x) \leq 1/4  $ for all $ x\in [0,x_0] $. 
Moreover, there  is some $ \alpha_0 \in( 0, 1/16] $ such that
\begin{equation}\label{eq:apriori1}
	 F(x_0) \ \leq \  \frac{x_0^{2\alpha_0} }{3} \, . 
\end{equation}
 We now define recursively $ x_n := x_{n-1}^2 $ for all $ n \in \mathbb{N} $. By induction on $ n $, one establishes that $ F(x_n) \leq \frac{1}{3} x_{n+1}^{\alpha_0}$ for all $ n \in \mathbb{N}_0 $. Namely, for $ n = 0 $ this is the content of~\eqref{eq:apriori1}. For the induction step, we use~\eqref{eq:boundcorse} which implies
\begin{equation}
	\frac{F(x_n)}{x_{n+1}^{\alpha_0}} \  \leq 3 \,  \max\left\{ \frac{F(x_{n-1})^K}{x_n^{2\alpha_0}} \, , \, C\,  x_{n}^{\frac{1}{4} - 2\alpha_0} \right\}   \ \leq \ 3 \, \max\left\{ \Big(\frac{F(x_{n-1})}{x_n^{\alpha_0}}\Big)^2 \, , \, C x_0^{\frac{1}{8}} \right\}  \leq \frac{1}{3} \, . 
\end{equation}
Since $ F $ is monotone increasing, this implies that for all $ x \in (x_{n+1} , x_n] $ and all $ n \in \mathbb{N} $ (and hence for all $ x \in (0,x_0] $):
\begin{equation}\label{eq:apriori2}
 F(x)  \ \leq \ F(x_n) \ \leq  \frac{x_{n+1}^{\alpha_0} }{3} \ \leq \  \frac{x^{\alpha_0} }{3}  \, . 
\end{equation}
This completes the proof of the initial (still insufficient) bound.

In the second step in the proof,  we will improve on the power law with which~\eqref{eq:apriori2} holds.   
To this end, suppose that for some $ \alpha ,  C > 0 $ and  all $ x\in (0,x_0] $:
\begin{equation}\label{eq:boundF2}
	F(x) \ \leq \ C \ x^\alpha 
\end{equation}
Then~\eqref{eq:nonlin} with $ y = x^{\frac{\alpha K}{1+3\alpha K} }$ implies that for all $ \alpha>0 $ and $ K \geq 2 $:
\begin{equation}
F(x) \ \leq  \ C \left( x^{\alpha K \frac{1+\alpha K}{1 + 3 \alpha K} } \, +\, x^{  \frac{r \alpha K}{1+3\alpha K} } \right)  \ \leq \ C \, \left( x^{\alpha  K \frac{1+\alpha K}{1 + 3 \alpha K} } +  \, x^{\frac{r}{3}} \right)\,  .
\end{equation}
with some constant $ C > 0 $. In case $ K \geq 3 $, this shows that one may improve the exponent $ \alpha $ in the bound~\eqref{eq:boundF2} by a factor larger than one as long as $ \alpha \leq  \frac{r}{3}$. This proves that the bound~\eqref{eq:boundF2} holds with $ \alpha = \frac{r}{3} $ if $ K \geq 3 $.

In case $ K = 2 $, we need to improve on the non-linear inequality~\eqref{eq:nonlin} in order to improve on the apriori bound.  To do so we will denote by $ x_{\pm 1} $ the two neighbors of the root and expand the recursion relation~\eqref{eq:recrel} one step further:
\begin{align}\label{eq:2step}
	\Im G(0,0;\zeta) \ &  \geq \ \, \sum_{\nu \in \{\pm 1\} } | G(0,x_\nu;\zeta) |^2 \sum_{v \in \mathcal{N}_{x_\nu} \backslash \{0\}}  \Im G^{\mathcal{T}_{x_\nu}}(v,v;\zeta) \notag \\
	& \geq \, \sum_{\nu \in \{\pm 1\} } \frac{| \Gamma(x_\nu)|^2}{( |V(0) | + |\zeta| +   | \Gamma(x_\nu)| + | \Gamma(x_{-\nu})|)^2}  \sum_{v \in \mathcal{N}_{x_\nu} \backslash \{0\}}  \Im G^{\mathcal{T}_{x_\nu}}(v,v;\zeta)  \, .
\end{align}
where  $  \Gamma(x_\nu) := G^{\mathcal{T}_0}(x_\nu,x_\nu;\zeta)  $. We pick $ y > 0 $ and condition on the event $ E_0 :=\left\{  |V(0) | \leq y^{-1} \right\} $.  We now distinguish three events, which decompose the probability space:
\begin{enumerate}
	\item $ E_1 := \left\{ |  \Gamma(x_\nu) | > y^{-1} \; \mbox{ for both $ \nu=\pm 1$} \right\} $.  
	\item $ E_2 := \left\{  | \Gamma(x_1)| > y^{-1} \; \mbox{and} \;  | \Gamma(x_{-1})| \leq y^{-1} \right\} $, or vice versa in case $ E_3 $.
	\item $ E_4 := \left\{ |  \Gamma(x_\nu) | \leq  y^{-1} \; \mbox{ for both $ \nu=\pm 1$} \right\} $.
\end{enumerate}
The probability of the event $ E_1 $ is estimated with the help of Lemma~\ref{lem:auxH}:
\begin{equation}
	\mathbb{P}\left(E_1\right)  \ = \ \mathbb{P}\left( |  \Gamma(x_{1}) |>  y^{-1}  \right)^2 \leq \ C y^2 \, F_\zeta(y)^4 \, . 
\end{equation}
In the event $ E_2 $ one of the quadratic factors in the first sum in~\eqref{eq:2step} is bounded from below,
\begin{equation}
	 \frac{| \Gamma(x_1)|}{( |V(0) | + |\zeta| +   | \Gamma(x_1)| + | \Gamma(x_{-1})|)} \geq  \frac{y^{-1}}{y^{-1} + |\zeta| + 2 y^{-1} } = \frac{1}{3 + y \, |\zeta| } \, . 
\end{equation}
Supposing that $ y $ is upper bounded, there is some constant $ C > 0 $ such that we may hence estimate 
\begin{equation}
	\mathbb{P}\left(\{ \Im G(0,0;\zeta) \leq x \} \cap E_0 \cap E_2  \right)  \  \leq \ \mathbb{P}\left(  \sum_{v \in \mathcal{N}_{x_1} \backslash \{0\}}  \Im G^{\mathcal{T}_{x_1}}(v,v;\zeta) \leq C x \right) \ \leq \ F_\zeta(C x)^2 \, . 
\end{equation}
In case $ E_0 \cap E_3 $ happens, one proceeds analogously. It therefore remains to estimate the last case in which we use~\eqref{eq:recrel} and the fact that $| G(0,0;\zeta) | \geq ( y^{-1} + |\zeta| + 2 y^{-1} ) $ in the event $ E_0 \cap E_4   $. There is hence some constant $ C > 0 $ such that
\begin{equation}
	\mathbb{P}\left( \{ \Im G(0,0;\zeta)\leq x \} \cap E_0 \cap E_4  \right)  \  \leq \ \mathbb{P}\left(  \sum_{v \in \mathcal{N}_{0} \backslash \{0\}}  \Im G^{\mathcal{T}_{0}}(v,v;\zeta) \leq  C x y^{-2}  \right) \ \leq \ F_\zeta( C x y^{-2}  )^2 \, . 
\end{equation}
Summarizing, we also have for all sufficiently small $ y > 0$: 
\begin{equation}
	F(x) \ \leq \ F( C x y^{-2}  )^2 + C \left( y^2 \, F(y)^4  + y^{r}   +  2 F(C x)^2 \right) \, . 
\end{equation}
We now proceed as in the case $ K \geq 3 $ and assume a bound of the form~\eqref{eq:boundF2}. Picking $ y = x^{\frac{\alpha}{1+4\alpha}} $ then yields
\begin{equation}
	F(x) \ \leq \ C \left( x^{\alpha \, \frac{2 + 4\alpha}{1+4\alpha}} + \, x^{\frac{r \alpha }{1+4\alpha} } +  x^{2\alpha}\right) \, . 
\end{equation}
Since $ r > 12 $ this yields an improvement of the exponent in~\eqref{eq:boundF2} provided $ \alpha \leq r/4 $. This completes the proof of the last assertion in Theorem~\ref{thm:Fprop} also for $ K = 2 $. 
\end{proof}

\medskip

\appendix

\section{On two notions of \emph{ac} spectrum} \label{App:ac}

As mentioned in the introduction, there is more than one natural choice for the notion  of the absolutely continuous spectrum of an ergodic self adjoint operator~$H$, and in particular of the spectral measure associated with the vector $\delta_0$.    In addition to $\sigma_{ac}(H)$ of \eqref{def:ac}, one may also consider  $\Sigma_{ac}(H)$ which is defined as the minimal closed set in which the $ac$ component of the spectral measure is supported (or equivalently: the collection of energies for which the $ac$ spectral measure of arbitrarily small neighborhoods is non zero).  
For better clarity in regard to the statements proven above, let us add  some general comments on the relation between the two sets 
$\sigma_{ac}(H)$ and $\Sigma_{ac}(H)$.

\begin{enumerate} 
\item It may well be that  random Schr\"odingier operators have only $ac$ and $pp$ spectra, with  $\sigma_{ac}(H)$ a finite union of open intervals and $\Sigma(H)$ being the disjoint union of $\Sigma_{pp}$ and $\sigma_{ac}$.   In that case, the only difference between $\sigma_{ac}(H)$ and $\Sigma_{ac}(H)$ is the  inclusion of the endpoints of the intervals.
However, the techniques available so far do not yet allow such a conclusion.  That includes  the criterion~\cite{AW11}  $\varphi(E;1)  < \log K$, which is expressed  in terms of a function whose continuity and other regularity properties are still not sufficiently understood.
\item  The measurable set $\sigma_{ac}(H)$ is dense within the closed set  $\Sigma_{ac}(H)$.  
\item  \label{it} Despite the above point it is not generally true, within the context of self adjoint operators, that the set difference between   $\Sigma_{ac}(H)$ and $ \sigma_{ac}(H)$ is  of measure zero.   (For a counterexample consider a self adjoint operator for which $\sigma_{ac}(H)$ is a union of a countable collection of intervals in $[0,1]$, the sum of whose measures is smaller than $1$, yet whose centers form a dense subset of  $[0,1]$.)  Hence, even though for every $E\in \sigma_{ac}(H)$:
\be  \label{def:ac2}
  \mathbb{P} \left( \Im G(0,0;E+i0) \neq 0 \right )  \ >\  0  \  \,.
\ee 
one \emph{cannot}  conclude that \eqref{def:ac2} holds also for almost every $E\in  \Sigma_{ac}(H)$. 
\item  Another relevant property which is not automatically shared by  $\Sigma_{ac}(H)$ is that within $\sigma_{ac}(H)$ the spectrum is (almost surely) \emph{purely} ac (the mean value of the singular measure of $\sigma_{ac}(H)$ is zero, though that is not known for $\Sigma_{ac}(H)$).   This follows by the arguments found in Refs.~\onlinecite{Aron,Sim_Wolff}, combined with the $0$-$1$ law which is valid for  tree graphs by which~\eqref{def:ac2} implies the stronger statement that the probability seen there is actually $1$. 
\end{enumerate} 
The point \ref{it}. is the main reason our results are formulated for $\sigma_{ac}(H)$ rather than $\Sigma_{ac}(H)$.   
\medskip 

\section{Proof of a (known) ballistic upper bound}\label{App:B}

For completeness of the presentation, following is a simple proof of a known ballistic upper bound which our main result complements.

\begin{proof}[Proof of \eqref{eq:ballistic1}]  
For an arbitrary graph $\G$, on which the distance is denoted by $d(x,y)$,
let $A$ be an operator on $\ell^2(\G)$ with kernel $A(x,y)$, 
for which
\be  
g(\alpha)  := \sup_{x\in \G} \sum_{y\in \G} |A(x,y)| \, e^{\alpha d(x,y)} \   <\  \infty 
\ee  for all 
$\alpha <  \infty$ (which implies that $ g $ is continuous).  For such operators the exponential function admits a convergent power series expansion.  Using it one finds, for each $x_0\in \G$: 
\be \label{exp1}
\sum_{x\in \G}\left|  \langle \delta_x \, , \, e^{-i t A } \, \delta_{x_0}\rangle \right| \, e^{\alpha d(x,y) } \  \le \   \sum_{n=0}^\infty \frac{t^n}{n!} 
\prod_{j=1}^n \sum_{x_j \in \G} \left[ |A(x_{j-1},x_j)| \, e^{\alpha d(x_{j-1}, x_j)} \right] 
  \  \le \  e^{t \, g(\alpha) }  \, ,   
\ee
where use was made of the triangle inequality: 
$d(x,x_0) \le \sum_{j=1}^n d(x_j,x_{j-1})$.

  Next, it may be noted that the above bound 
  extends unchanged to operators of the form $H= A+V$, with $A$ as above and $V$ an arbitrary real valued potential, that is a self adjoint operator with a diagonal  kernel.  
 Applying the  bound \eqref{exp1} within the Lie- Trotter formula 
  \be 
 \langle \delta_x \, , \, e^{-i t H } \, \delta_{x_0}\rangle  \  = \  \lim_{k\to \infty}\;   \langle \delta_x \, , \,  \left[ e^{-i A t /k } e^{-i V t/k}\right]^k \delta_{x_0}\rangle
\ee 
we get: 
\be \label{exp2}
\sum_{x\in \G} \left|  \langle \delta_x \, , \, e^{-i t H } \, \delta_{x_0}\rangle \right| \, e^{\alpha d(x,y) } \  \le \   
\lim_{k\to \infty}  \left[e^{ g(\alpha) t/k }\right]^k   \ = \ 
 \  e^{t \, g(\alpha) }  \, . 
\ee 
Let now $\hat v := \min_{\alpha >0} g(\alpha)/\alpha$ and denote by $\mu \in (0,\infty) $ the (largest) minimizing $\alpha$   (both are easily seen to be well defined).   Applying the Chebyshev inequality to the probability of the event and using the fact that $ \left|  \langle \delta_x \, , \, e^{i t H } \, \delta_{x_0}\rangle \right|  \leq 1 $ one gets, for any $v> \hat v$: 
\begin{eqnarray}  
{\rm Pr}_{\delta_0,t}(|x|> v t) \ &\le& \  e^{- \mu t   v} \sum_{x\in \G}   {\rm Pr}_{\delta_0,t}(x)   \, e^{\mu |x|}  \ \leq \  e^{- \mu t   v} \sum_{x\in \G}  
 \left|  \langle \delta_x \, , \, e^{-i t H } \, \delta_{0}\rangle \right|  \, e^{\mu |x|}  \nonumber 
  \\[2ex]   &\le&  \ e^{- \mu t v} e^{t \, g(\mu) }  \ =\   e^{ -  \mu t  [v - \hat v]  }  \, , \\  \nonumber 
\end{eqnarray} 
where $|x|:= d(x,0)$.  This directly implies the ballistic upper bound \eqref{eq:ballistic1} for the operator $H$ of \eqref{def:H}.    

\end{proof}

Let us note that a generalization of the upper bound \eqref{eq:ballistic1}  
(though not of the lower bound which is our main result)
 can be found in the Lieb-Robinson theorem on the finiteness of the speed of sound in  many body quantum systems~\cite{LR}, a general result which is of current interest.\cite{NS}  In fact, Eq.~\eqref{eq:ballistic1} can also be concluded from the latter through the second quantization formalism in the context of a system of non-interacting Fermions.       \\

\noindent\textit{Acknowledgments.}
We are grateful  for the NSF support of the research presented here, which was supported in part by grants PHY-1104596  (MA) and DMS-0701181 (SW).

\end{document}